\documentclass[11pt]{amsart}

\issueinfo{00}{0}{Month}{Year}
\copyrightinfo{2008}{University of Calgary}
\pagespan{1}{2}
\PII{ISSN 1715-0868}
\usepackage{graphicx}
\usepackage{epstopdf}
\usepackage{amsthm,amsmath,amssymb,amsxtra,amscd}
\usepackage{verbatim}
\usepackage{indent}
\usepackage{url}
\newtheorem{theorem}{Theorem}[section]

\newtheorem{lemma}[theorem]{Lemma}
\newtheorem{corollary}[theorem]{Corollary}

\newtheorem{proposition}[theorem]{Proposition}

\newtheorem*{example*}{Example}

\newtheoremstyle{myexample}{3pt}{3pt}{\rmfamily}{}{\itshape}{:}{ }{\thmname{#1}\thmnumber{ #2}\thmnote{ (#3)}}
\theoremstyle{myexample}
\newtheorem{numremark}[theorem]{Remark}

\newtheoremstyle{myremark}{3pt}{3pt}{\rmfamily}{}{\itshape}{:}{ }{\thmname{#1}}
\theoremstyle{myremark}

\newtheorem*{observation*}{Observation}

\newtheoremstyle{conjecture}{3pt}{3pt}{\itshape}{}{\bfseries}{.}{ }{\thmname{#1}\thmnote{ (#3)}}
\theoremstyle{conjecture}

\newtheorem*{question*}{Question}

\newtheorem{theorem*}{Theorem}

\numberwithin{equation}{section}

\newcounter{algorithm}
\renewcommand{\thealgorithm}{\thesection.\arabic{algorithm}}
\newenvironment{algorithm}[1]{%
  \null
  \refstepcounter{algorithm}%
  \hrule%
  \vspace{0.2em}%
  \noindent\textbf{Algorithm \thealgorithm} #1
  \vspace{0.2em}%
  \hrule%
  \vspace{0.2em}
  }{%
  \vspace{0.2em}%
  \hrule%
  \null
}

\begin{document}

% A short title is not required, but if needed use:
% \title[short title]{full title}
\title{Geometric clustering in normed planes}

\author{Pedro Mart\'{\i}n}
\address{Departamento de Matem\'{a}ticas, UEx, 06006 Badajoz, Spain}
\email{pjimenez@unex.es}

\author{Diego Y\'{a}\~{n}ez}
\address{Departamento de Matem\'{a}ticas, UEx, 06006 Badajoz, Spain}
\email{dyanez@unex.es}

\date{(date1), and in revised form (date2).}
\subjclass[2000]{46B20, 52A10, 52A21, 52B55, 65D18}
\keywords{geometric clustering, normed space}

\thanks{}

	\begin{abstract}
		Given two sets of points $A$ and $B$  in a normed plane, we prove that there are two linearly separable sets $A'$ and $B'$ such that $\mathrm{diam}(A')\leq \mathrm{diam}(A)$, $\mathrm{diam}(B')\leq \mathrm{diam}(B)$, and $A'\cup B'=A\cup B.$ This extends a result for the Euclidean distance to symmetric convex distance functions. As a consequence, some Euclidean $k$-clustering algorithms are adapted to normed planes, for instance, those that minimize the maximum, the sum, or the sum of squares of the $k$ cluster diameters. The 2-clustering problem when two different bounds are imposed to the diameters is also solved. The Hershberger-Suri's data structure for managing ball hulls can be useful in this context.

	\end{abstract}

\maketitle

\section{Introduction and notation}
Given a set $S$ of $n$ points in the plane, a \textit{cluster} is any nonempty subset of $S$, and a $k$-\textit{clustering} is a set of $k$ clusters such that any point of $S$ belongs to some cluster.
Fixed a distace function on the plane, in general, a \textit{clustering problem} asks for a $k$-clustering of $S$ that minizes or maximizes a function $\mathcal{F}:\mathbb{R}^k\to \mathbb{R}$ defined on the clusters, where usually $\mathcal{F}$ depends on the distance function.
%Given a set $S$ of $n$ points in the plane, the problem of finding a fixed number $k$ of clusters of $S$ had been studied by a lot of authors under different conditions. If $k=1$, the radius of the minimal enclosing discs is studied in the Euclidean case (\cite{EH}, \cite{Sh-H}) and extended to a general normed plane (\cite{Al-Ma-Sp}, \cite{Al-Ma-Sp2}, \cite{J}).
For instance,   Avis (\cite{Av}, $O(n^2 \log n)$ time) and Asano et al. (\cite{Asano}, $O(n \log n)$ time) for $k=2$, and  Hagauer and Rote  (\cite{Hagauer-Rote}, $O(n^2 \log^2 n)$ time) for $k=3$, present algorithms that minimize the maximum Euclidean diameter of the clusters.
 %Let $F:\mathbb{R}^k\rightarrow \mathbb{R}$ be a monotone increasing function (for example, $\mathcal{F}$ can be the maximum, the sum, or the sum of squares of $k$ non-negative arguments), %Given a monotone increasing function with respect to the diameter or to the radius $\mathcal{F}$,
 Capoyleas et al. (\cite{Capoyleas and others})  prove that if $\mathcal{F}$ is a monotone increasing function applied over the diameters or over the radii of the clusters in the Euclidean plane, the $k$-clustering problem of minimizing    $\mathcal{F}$ can be solved in polynomial time. Examples of $\mathcal{F}$ are the maximum, the sum, or the sum of squares of $k$ non-negative arguments. All the algorithms cited above are based on the fact that any two clusters in an optimal solution can be separated by a line. We prove in Section \ref{linear separability of clusters}
 that this last statement is true for any symmetric convex distance function (Theorem \ref{separacion de dos conjuntos2}), and as a consequence we justify in Section \ref{2-clustering problems}, Section \ref{$k$-clustering problems}, and Section \ref{3-clustering} %Section \ref{some applications}
 that all such as approaches work correctly in every normed plane.

Hershberger and Suri (\cite{H-S}) consider the 2-clustering problem where individual constraints are specified for each of the clusters. Given a measure $\mu$, and a pair of positives real numbers $d_1$ and $d_2$, they find algorithms to split $S$  into two subsets $S_1$ and $S_2$ such that $\mu(S_1)\leq d_1$ and $\mu(S_2)\leq d_2$. The measure $\mu$ can be the Euclidean diameter of the set ($O(n \log n)$ time); the area, perimeter, and diagonal of the smallest rectangle with sides parallel to the coordinates axes ($O(n \log n)$ time); or the radius of the  smallest enclosing sphere with the norms $L_1$ ($O(n \log n)$ time) and $L_2$ ($O(n^2 \log n)$ time). Although we prove  that Hersberger-Suri's approach does not work for every normed plane when $\mu$ is the diameter, an optimal solution based on separable sets can always be computed in $O(n^2 \log n)$ time (Section \ref{2-clustering constraints}).
% When our objective is this last one of minimizing the radius of two sets containing $S$, we face to the 2-center problem, that has been widely studied (see \cite{Ag-Sh}, \cite{Ag-Sh-SWe}, \cite{Ag-Pa-Av-Ri-Sh}, \cite{Chan}, \cite{Ep}, \cite{H-S-G}, \cite{H2}, \cite{J-K},  and \cite{Sh}).

In order to solve the above Euclidean 2-clustering problem,
Hershberger and Suri introduce %in \cite{H-S}
 a data structure that stores the information about the intersection set of all the balls of a given radius $d$ that contain $S$, usually called $d$-\textit{ball hull} or $d$-\textit{circular hull} of $S$. This data structure is an interesting tool for other $k$-clustering algorithms in the Euclidean subcase, and can play an important roll when others norms are considered (see Appendix). For instance it is useful in  the extension of Hagauer-Rote's algorithm (Section \ref{3-clustering}).% and \cite{Ma-Ma-Sp2}).

From now on, we denote by $\mathbb{E}^2$ the Euclidean plane, and by $\mathbb{M}^2$  a \textit{normed plane}, namely, $\mathbb{R}^2$ endowed with a convex symmetric distance funtion $\|\cdot\|$. We call $B(x,r)$ to the \textit{ball with center $x\in \mathbb{M}^2$ and radius $r>0$}, and $S(x,r)$ to the \textit{sphere} of $B(x,r)$. We use the usual abrevations $\mathrm{diam}(A)$ and  $\mathrm{conv}(A)$ for the \textit{diameter} and the \textit{convex hull} of a set $A$, $\overline{ab}$ for the \textit{line segment} meeting two points $a,b\in \mathbb{M}^2$, and  $\langle a,b \rangle$  for its affine hull.
%The rest of the paper is organized in two sections. Section \ref{linear separability of clusters} is devoted to prove the linear separability of the optimal solutions for some $k$-clusterings problems in normed planes. In Section \ref{some applications}, using the above result, we justify that all the approaches of Avis, of Assano et al., of Capoyleas et al., and of Hagauer and Rote work correctly in such as normed planes. We make use of the data structure of Hershberger and Suri for this last algorithm.

\section{Linear separability of clusters}\label{linear separability of clusters}

We say that two sets of points in  $\mathbb{M}^2$ are \emph{linearly separable} (for short, \emph{separable}) if there exits a line $L$ such that every set is situated in a different closed half plane defined by $L$. The following result is presented in \cite{Capoyleas and others}.

\begin{theorem}\label{separacion de dos conjuntos}
	Let $A$ and $B$ be two sets of  a finite number of points in $\mathbb{E}^2$. Then, there are two separable sets $A'$ and $B'$ such that $\mathrm{diam}(A')\leq \mathrm{diam}(A)$, $\mathrm{diam}(B')\leq \mathrm{diam}(B)$, and $A'\cup B'=A\cup B.$
\end{theorem}

%With this result, the authors prove that for any fixed $k$, it is possible to find a $k$-clustering which minimizes any monotone function of the diameters (or the radii) of the clusters in polynomial time. The same result is used in \cite{Hagauer-Rote} in order to prove that any set of $n$ points in the Euclidean plane can be separated into three classes such that the maximum distance between two points in the same class is minimized in $O(n^2 \log^2 n)$ time.

%From now on and without loss of generality, we can assume that  $\mathrm{diam}(A)\geq \mathrm{diam}(B)$.
%Objective: either to prove Theorem \ref{separacion de dos conjuntos} or to find a counterexample for the statement in a general normed plane.

In the rest of this section we work in $\mathbb{M}^2$ and our objective is to prove the statement of Theorem \ref{separacion de dos conjuntos}. Without loss of generality we assume that  $\mathrm{diam}(A)\geq \mathrm{diam}(B)$. Let us denote  $\{u_1,u_2,\dots, u_{2k}\}$ the sequence of points in clockwise order where the boundaries of $\mathrm{conv}(A)$ and $\mathrm{conv}(B)$ cross (Figure \ref{convexAB}). $\mathrm{conv}(A)\setminus\mathrm{conv}(B)$ and $\mathrm{conv}(B)\setminus\mathrm{conv}(A)$ are made by two interlacing sequences of polygons $\{A_1,A_2,\dots,A_k\}$ and $\{B_1,B_2,\dots,B_k\}$ such that (for convenience, $u_{2k+1}:=u_1$ and $A_{k+1}:=A_1$):  $A_i$ touches $B_i$ at $u_{2i}$; $B_i$ touches $A_{i+1}$ at $u_{2i+1}$; the vertices of any $A_i$ belong either to $A\setminus B$ or to $\mathrm{conv}(A)\cap \mathrm{conv}(B)$; the vertices of any $B_j$ belong either to $B\setminus A$ or to $\mathrm{conv}(A)\cap \mathrm{conv}(B)$. We say that $(A_i,B_j)$ is a \emph{bad pair} if $\mathrm{diam}(A_i\cup B_j)>\mathrm{diam}(A)$. In such as case, $A_i$ is a \emph{bad set} and $B_j$ is its \textit{bad partner}, and viceversa. If $\|a_i-b_j\|>\mathrm{diam}(A)$ for some $a_i\in A_i$ and $b_j\in B_j$, then both $a_i$ and $b_j$ are \textit{bad points}, $a_i$ is a \textit{bad partner} of $b_j$ (and viceversa), and the segment $\overline{a_ib_j}$ is a \textit{bad segment}.
%Two bad pairs $(A_i,B_j)$ and $(A_{i'}, B_{j'})$ cross if and only if every two bad segments connecting a bad point in $A_i$ and a bad partner in $B_j$, and a bad point in $A_{i'}$ to a bad partner in $B_{j'}$ intersect, independent of the choice of these four bad points. %Such segments are called \emph{bad segments}.

\begin{figure}
	\vspace*{-1cm}
\begin{center}
	\hspace*{-1cm}
%\scalebox{0.5}{\includegraphics*[58, 144][755, 495]{dosdefinitivoapaisado2.ps}}\\
\scalebox{0.65}{\includegraphics[width=17cm,angle=270]{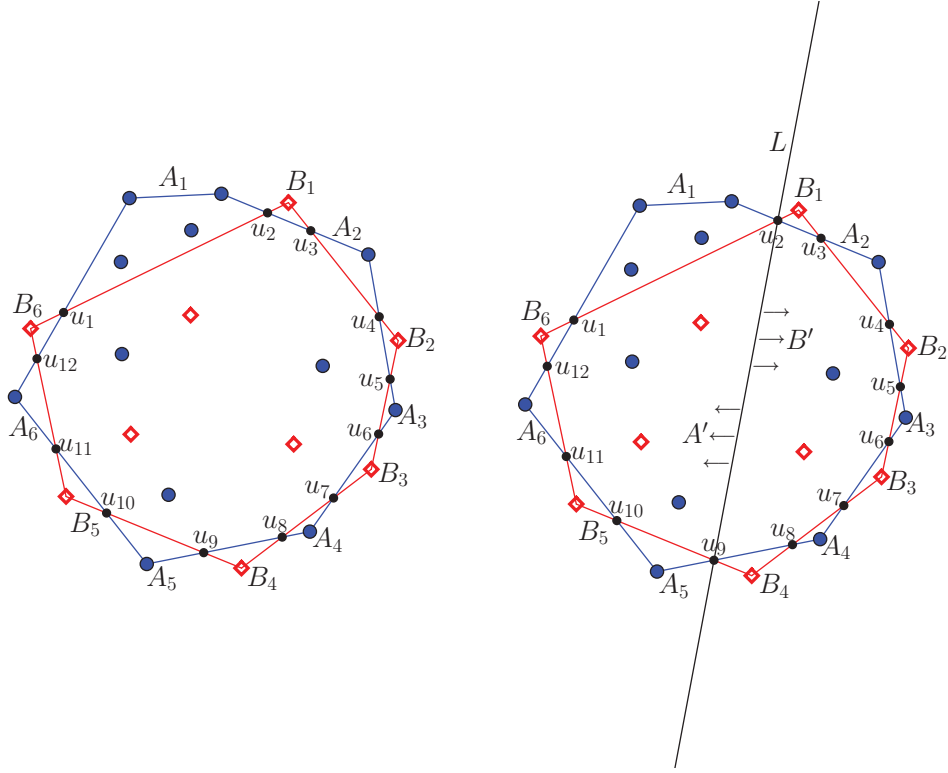}}\\
\vspace{-0.75cm}\caption{$A$ (blue points) and $B$ (red points) are not separable (left). $A\cup B$ can be split by $L$ into new subsets $A'$ and $B'$ without increase of the Euclidean diameters (right).}\label{convexAB}
	\end{center}
\end{figure}

%\begin{figure}[ht]
%	\begin{center}
%		\vspace*{-1cm}
%		\hspace*{-4cm}\includegraphics[width=15cm]{convexAB.pdf}\\
%		\vspace*{-2cm}\caption{$\mathrm{conv}(A)\cap \mathrm{conv}(B)$}\label{convexAB.pdf}
%	\end{center}
%\end{figure}

%We can prove the following.
\begin{lemma}\label{cross0}
 Let $(A_i,B_j)$ and $(A_{i'}, B_{j'})$ be two bad pairs such that $A_i\neq A_{i'}$ and $B_j\neq B_{j'}$. Let us choose $a_i\in A_i, b_j\in B_j, a_{i'}\in A_{i'}, b_{j'}\in B_{j'}$ such that $\overline{a_ib_j}$ and $\overline{a_{i'}b_{j'}}$ are bad segments. Then, either these bad segments intersect, or any point $a\in A_m$ belonging to the halfplane defined by $\langle b_jb_{j'}\rangle$ where $a_i$ and $a_{i'}$ are not contained, is not bad.
\end{lemma}
\begin{proof}
	Let us assume that $\overline{a_ib_j}$ and $\overline{a_{i'}b_{j'}}$ are bad segments with an empty intersection set.
	There are two cases (disregarding symmetric variations) for the relative positions of the points on the boundary of $\mathrm{conv}(\{a_i,a_{i'},b_j,b_{j'}\})$.
	
	Case 1: \textit{$a_i,b_{j'},a_{i'},b_j$ is the sequence of the points in clockwise order}. Then, we get a contradiction:
	\begin{multline*}
	\mathrm{diam}(A)+\mathrm{diam}(B)\geq \|a_i-a_{i'}\|+\|b_j-b_{j'}\|\geq\\ \|a_i-b_j\|+\|a_{i'}-b_{j'}\|>2\ \mathrm{diam}(A).
	\end{multline*}
		
	Case 2: \textit{$a_i,a_{i'},b_{j'}, b_j$ is the sequence of the points in clockwise order}. Let us assume that there exists a bad segment $\overline{a_mb_k}$ such that $a_m\in A_m$ belonging to the halfplane defined by $\langle b_jb_{j'}\rangle$ where $a_i$ and $a_{i'}$ are not contained. The half-lines starting in $a_m$ and connecting $a_m$ with  $a_i$ and with $a_{i'}$, and the lines $\langle a_m,b_j \rangle$  and $\langle a_m,b_{j'} \rangle$, divide the plane in six zones (see Figure \ref{posicion4puntosmalos}).
		\begin{figure}[ht]
		\begin{center}
%			\vspace*{-2cm}
%			\hspace*{-0.5cm}
			\scalebox{0.5}{\includegraphics{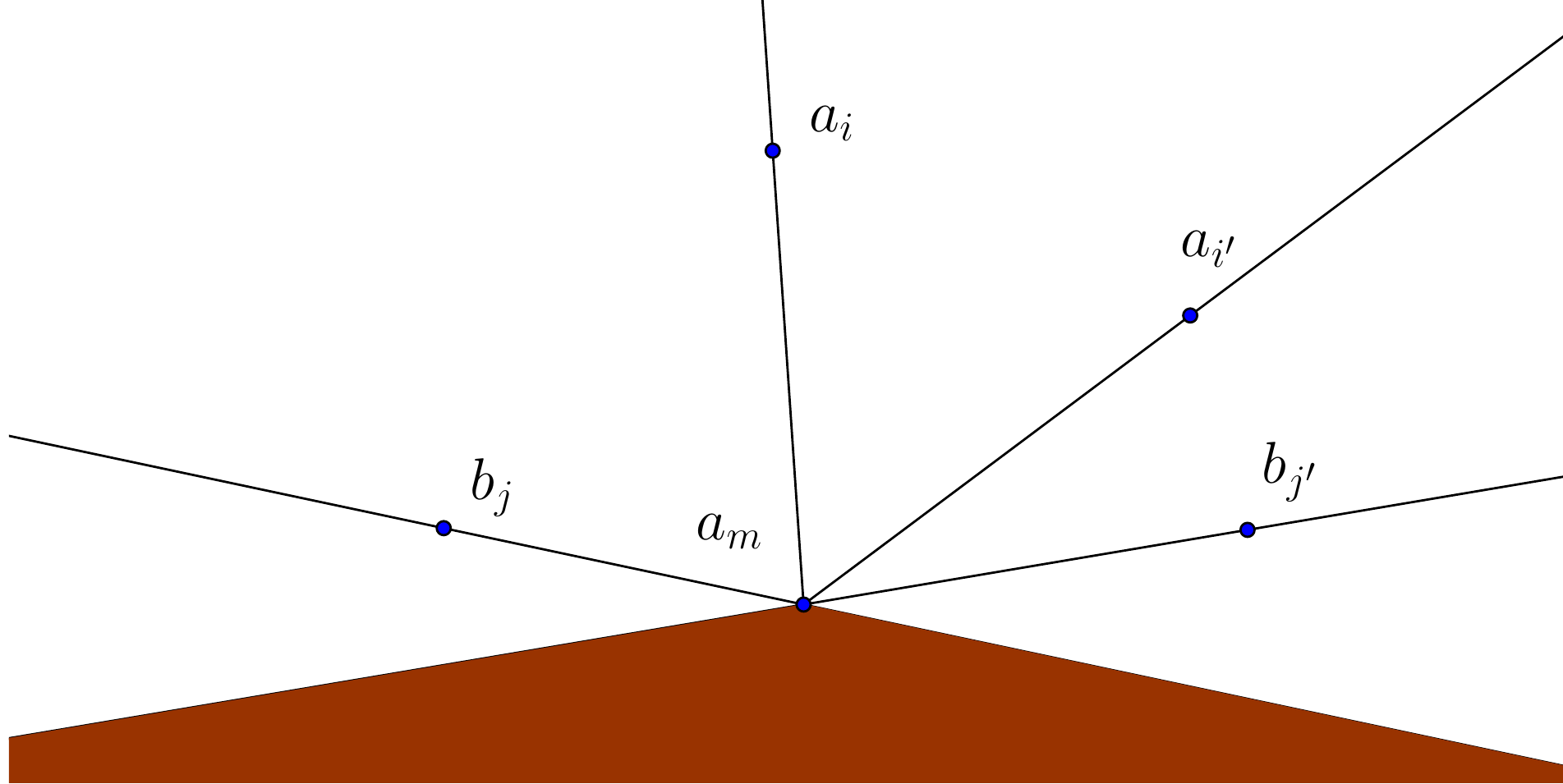}}\\
%			\vspace*{-1.5cm}
\caption{If $(a_i,b_j)$ and $(a_{i'},b_{j'})$ are bad partners, then the shaded zone can not contain a bad partner of $a_m\in A_m$}\label{posicion4puntosmalos}
		\end{center}
	\end{figure}
	%
%		\begin{figure}[ht]
%		\begin{center}
%			\vspace*{-2cm}
%			\hspace*{-0.5cm}
%			\scalebox{0.2}{\includegraphics[width=13cm]{posicion4puntosmalosconletras.pdf}}\\
%			\vspace*{-1.5cm}\caption{If $(a_i,b_j)$ and $(a_{i'},b_{j'})$ are bad partners, then the shaded zone can not contain a bad partner of $a_m\in A_m$}\label{posicion4puntosmalos}
%		\end{center}
%	\end{figure}
		By convexity, one of these zones (the shaded zone in Figure \ref{posicion4puntosmalos}) can not contain $b_k$. If $b_k$ belongs to whichever other zone, it is possible to consider a quadrangle whose vertices are situated in clockwise order like in Case 1, and we get a contradiction. Therefore, if this case holds, then $a_m$ is not a bad point.
\end{proof}

%\begin{remark}
%		Case 2 in Lemma \ref{cross0} implies that $\|a_i-b_{j'}\|$ or $\|a_{i'}-b_j\|$ are greater than $\mathrm{diam}(A)$, because the sum of the diagonals of the quadrangle $a_i,a_{i'},b_{j'}, b_j$ must be greater than or equal to the sum of two opposite sides.
%\end{remark}

\begin{numremark}
	%In Any two disjoint bad pairs $(A_i,B_j)$, $(A_{i'},B_{j'})$ cross in the Euclidean plane, and for the proof
		Case 2 does not occur in $\mathbb{E}^2$, and as a consequence every two bad segments from disjoint bad pairs $(A_i,B_j)$ and $(A_{i'},B_{j'})$ cross. In order to prove this in \cite{Capoyleas and others}, it is used the property that in an obtuse triangle the longest side  is opposite to the obtuse angle. %; and $(2)$, the quadrangle $a_i,a_{i'},b_{j'}, b_j$ has and obtuse angle and $such that, $\overline{a_ib_j}$ and $\overline{a_{i'}b_{j'}}$ are bad segments, then the clockwise order of the points  can not be  $a_i,a_{i'},b_{j'},b_j$.
		But if we consider the normed plane with unit sphere made by two arcs of circunferences showed in Figure \ref{contraej_2_configuracion}, the triangle with vertices $a_m,b,c,$ has an obtuse angle on vertix $a_m$, and the side $\overline{bc}$ is not the longest one. Besides, % is greater than 90 degree, %If we consider the sets $A=\{a_i,a_{i'},a_m\}$ and $B=\{b_j,b_{j'}\}$, then
	 there is a configuration of points similar to Case 2 where $\overline{a_{i'}b_{j'}}$ and $\overline{a_{i}b_{j}}$ are non intersecting, and such that $\mathrm{min}\{\|a_i-b_j\|, \|a_{i'}-b_{j'}\|\}>\mathrm{diam}(\{a_i,a_{i'},a_m\})=1>\mathrm{diam}(\{b_j,b_{j'}\})$.

\end{numremark}

%\begin{figure}[ht]
%		\vspace*{-2cm}
%		\hspace*{-1.3cm}
%		\scalebox{.8}{\includegraphics[width=14cm, angle=270]{contraej_2_4.pdf}}
%        %\includegraphics[width=6cm, angle=270]{contraej_2_4_2.eps}
%				\vspace{-2cm}
%		\caption{$\|b-c\|<\|a_m-b\|=\|a_m-c\|=1$ and \hspace{2cm}
%$	\|a_{i'}-b_{j'}\|>\|a_i-b_j\|>\mathrm{diam}(\{a_i,a_{i'},a_m\}) >\mathrm{diam}(\{b_j,b_{j'}\})$.}\label{contraej_2_configuracion}
%	\end{figure}

	\begin{figure}[ht]
		\vspace*{-1cm}
    \hspace*{-1.5cm}
%\begin{center}		
%\scalebox{.8}{\includegraphics[width=14cm, angle=270]{contraej_2_4.pdf}}
\includegraphics[width=15cm]{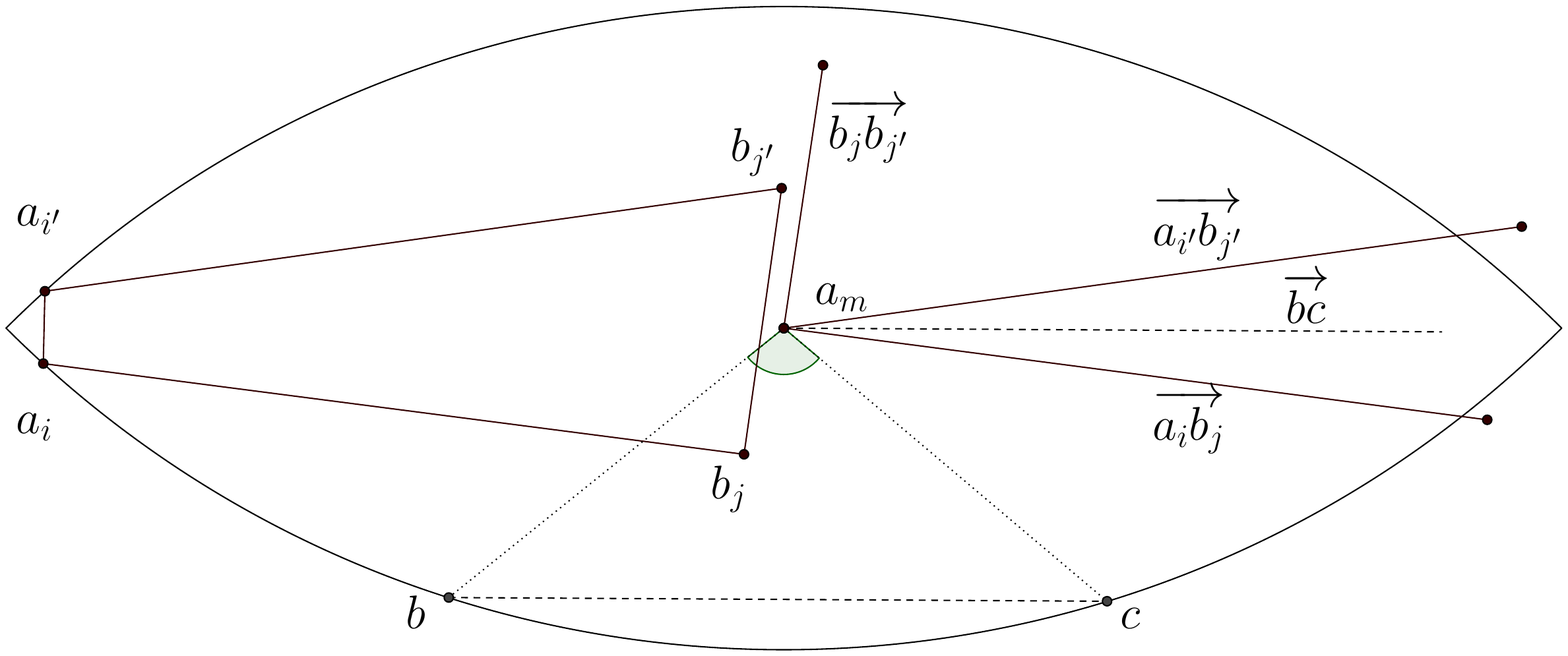}
\vspace{-2cm}\caption{$\|b-c\|<\|a_m-b\|=\|a_m-c\|=1$ and \hspace{2cm}
$	\|a_{i'}-b_{j'}\|>\|a_i-b_j\|>\mathrm{diam}(\{a_i,a_{i'},a_m\}) >\mathrm{diam}(\{b_j,b_{j'}\})$.}\label{contraej_2_configuracion}
%\end{center}	
\end{figure}

%\bigskip

Before splitting the sets $A$ and $B$, we group all the bad adjacent subsets $A_i$ from the cluster $A$. Namely, maximal cyclic groups of bad subsets $A_i$ are made. If $A_i$ and $A_{i'}$ (clockwise order) are bad subsets belonging to the same group, then there is not any not bad $A_k$  between   $A_i$ and $A_{i'}$, although some not bad $B_j$ can be situated between $A_i$ and $A_{i'}$. The same is made with cluster $B$. These maximal cyclic groups are noted by $\bar{A}_1, \bar{A}_2,\dots,\bar{A_p}$ and $\bar{B_1},\bar{B_2},\dots,\bar{B_q}$.

We say that $(\bar{A}_i, \bar{B}_{j})$ is a \textit{bad pair of groups} if there exits a bad segment from $\bar{A}_i$ to $\bar{B}_{j}$. Two pair of sets $(A_i, B_{j})$ and  $(A_{i'}, B_{j'})$ cross if there exist two (one from every pair) bad-crossing segments. Similarly, $(\bar{A}_i, \bar{B}_{j})$ and  $(\bar{A}_{i'}, \bar{B}_{j'})$ cross if there exist two (one from every pair) bad-crossing segments.
\begin{lemma}\label{group}
	Let $(A_i,B_j)$ and $(A_{i'},B_{j'})$ be two  bad pairs such that $A_i\neq A_{i'}$ and $B_j\neq B_{j'}$. If $A_i$ and $A_{i'}$ belong to a group $\bar{A_k}$, then a group $\bar{B_t}$ contains $B_{j}$ and $B_{j'}$.
\end{lemma}
\begin{proof}
	Let us assume that $(A_i,B_j)$ and $(A_{i'},B_{j'})$ are two bad pairs such that $A_i$ and $A_{i'}$ belong to the same group, but $B_j$ and $B_{j'}$ belong to different groups. Then it must exist a bad set $A_m$ between $B_j$ and $B_{j'}$. Let $B_k$ be a bad partner of $A_m$. $(A_i,B_j)$ and $(A_{i'},B_{j'})$ must cross (if not, by Lemma \ref{cross0}, $A_m$ can not contain bad points). Since $A_i$ and $A_{i'}$ belong to the same group, only one of them (not both) cross with $(A_m,B_k)$. Let us assume that  $(A_i,B_j)$ and $(A_m,B_k)$ cross. There exist  $a_m\in A_m,b_k\in B_{k}, a_{i'}\in A_{i'},b_{j'} \in B_{j'}$ that would be situated in an impossible  clockwise order $a_m,b_{k},a_{i'},b_{j'}$ (similar to Case 1 in Lemma \ref{cross0}), and we get a contradiction.
\end{proof}

Due to Lemma \ref{group}, the number of maximal cyclic groups for $A$ and for $B$ is the same.

\begin{lemma}\label{cross}
	Let $(\bar{A}_i, \bar{B}_{j})$ and $(\bar{A}_{i'}, \bar{B}_{j'})$ be two bad pair of groups such that $\bar{A}_i\neq \bar{A}_{i'}$ and $\bar{B}_j\neq \bar{B}_{j'}$. Then $(\bar{A}_i, \bar{B}_{j})$ and $(\bar{A}_{i'}, \bar{B}_{j'})$ cross.
\end{lemma}

\begin{proof}
%	Let us assume that $(\bar{A}_i, \bar{B}_{j})$ and $(\bar{A}_{i'}, \bar{B}_{j'})$ are two disjoint (namely, $\bar{A}_i\neq \bar{A}_{i'}$, $\bar{B}_j\neq \bar{B}_{j'}$) bad pairs of groups.
The clockwise order can not be $\bar{A}_i, \bar{B}_{j'},\bar{A}_{i'}, \bar{B}_{j}$ (due to the arguments used in Lemma \ref{cross0}, Case 1); and neither $\bar{A}_i, \bar{A}_{i'}, \bar{B}_{j'},\bar{B}_{j}$, because then  $\bar{B}_{j'}$ and $\bar{B}_{j}$ can not be separated by a bad polygon $A_m$ (Lemmma \ref{cross0}, Case 2). Therefore, the clockwise order must be $\bar{A}_i, \bar{A}_{i'}, \bar{B}_{j}, \bar{B}_{j'}$, and the groups cross.
\end{proof}
And we obtain the following from Lemma \ref{cross}.
\begin{corollary}
	There is an odd number of groups from each cluster, and they are completely interlacing.
\end{corollary}

%From this point forward, the approach of Capoyleas et al. can be followed.
Let $A_i$ be the last bad set of a group (in clockwise order), and let $B_{j'}$ be the last bad partner of $A_i$. Let $B_j$ be the first bad set after $A_i$, and let $A_{i'}$ be the first bad partner of $B_j$. We choose the separating line $L$ to go through the point $u_{2j}$ before $B_j$ and the point $u_{2j'+1}$ after $B_{j'}$ (see Figure \ref{convexAB}). We define $B'$ to be the points in $A\cup B$ lying on the same side of $L$ as $B_j$ and $B_{j'}$, and $A'$ as the remaining points.

\begin{proposition}
	The diameter of $A'$ is less than or equal to the diameter of $A$.
\end{proposition}
\begin{proof}
	Since $L$ cuts all bad pairs, there does not exist a bad point $a'\in A'$ with a bad partner inside $A'$ (the same happens with $B'$), and the diameter of  $A$ (and as well as the diameter of $B'$) have length less than or equal to $\mathrm{diam}(A)$.	
\end{proof}	

\begin{proposition}
	The diameter of $B'$ is less than or equal to the diameter of $B$.
\end{proposition}
\begin{proof}
Let $a,b\in B'$. We have to prove that $\|a-b\|\leq \mathrm{diam}(B).$   % We have to prove that $\|a-b\|\leq \mathrm{diam}(B)$.
	 If $a,b\in B$ there is nothing to prove. In other case, let us assume that $\|a-b\|>\mathrm{diam}(B)$. Let us choose $a_i\in A_i$, $b_j\in B_j$, $a_{i'}\in A_{i'}$, $b_{j'}\in B_{j'}$ such that $(a_i,b_{j'})$ and $(a_{i'},b_{j})$ are bad pairs.
	There are three possible cases.
	
	Case 1: $a\in \mathrm{conv}(A)\setminus \mathrm{conv}(B)$ and $b\in \mathrm{conv}(B)\setminus \mathrm{conv}(A)$. The points $\{b_j,a,b,b_{j'},a_{i'},a_i\}$ are situated around $\mathrm{conv}(A)\cap \mathrm{conv}(B)$ and it is possible to consider a clockwise order.   	
 If $\{a,b\}$ is the clockwise order of these two points, we observe the quadrangle with vertices (clockwise) $\{b_j,a,b,a_{i'}\}$ and the following contradiction holds:
	\begin{multline}\label{ineq}
	\mathrm{diam}(A)+\mathrm{diam}(B)\geq \|a-a_{i'}\|+\|b-b_j\|\geq\\ \|b_j-a_{i'}\|+\|a-b\|>
	\mathrm{diam}(A)+\mathrm{diam}(B).
	\end{multline}
	If the clockwise order is $\{b,a\}$, we obtain a similar contradiction on the quadrangle with vertices (clockwise order) $\{a_i,b,a,b_{j'}\}$.

	Case 2: $a,b\in \mathrm{conv}(A)\setminus \mathrm{conv}(B)$.  Case 1 implies that $\|b-b'\|\leq \mathrm{diam}(B)$ for every $b'\in (\mathrm{conv}(B)\setminus \mathrm{conv}(A))\cap B'$. %$a$ and $b$ are inside the region $\mathrm{conv}(A_{i'},A-B)\ \cap\ B'$ (the shaded zone in Figure \ref{contraej_5_1_conletra}), and
	If $\{a,b\}$ is the clockwise order of these two vertices, we can apply an argument similar to (\ref{ineq}) to the quadrangle $\{b_j,a,b,a_{i'}\}$:
	\begin{multline*}
	\mathrm{diam}(A)+\mathrm{diam}(B)\geq \|a-a_{i'}\|+\|b-b_j\|\geq\\ \|b_j-a_{i'}\|+\|a-b\|>
	\mathrm{diam}(A)+\mathrm{diam}(B),
	\end{multline*}
	which is again a contradiction. If the order is $\{b,a\}$, we use the quadrangle $\{b_j,b,a,a_{i'}\}$.

%\begin{figure}[ht]
%		\begin{center}
%			\scalebox{.5}{\includegraphics{contraej_5_1_conletra_wedt.eps}}\\
%\caption{}\label{contraej_5_1_conletra}
%		\end{center}
%	\end{figure}

	Case 3: $a\in \mathrm{conv}(A)\setminus \mathrm{conv}(B)$ and $b\in \mathrm{conv}(A) \cap \mathrm{conv}(B)$. %$b\in A\cap B$. %
	%The proof in \cite{Capoyleas and others} is valued:
	%We may assume that $b$ is a vertex of  $\mathrm{conv}(A) \cap \mathrm{conv}(B)\cap L^+,$ where $L^+$ denotes the halfplane that contains $B'$,
	Since the distance from $a$ is maximized at some vertex of $\mathrm{conv}(A) \cap \mathrm{conv}(B)\cap \mathrm{conv}(B'),$ we may assume that $b$ is one of these vertices and apply an analysis similar to Case 1 or  to Case 2.
\end{proof}

Using the previous results, we obtain the main theorem.

\begin{theorem}\label{separacion de dos conjuntos2}
	Let $A$ and $B$ be two sets of a finite number of points in $\mathbb{M}^2$. Then, there are two linearly separable sets $A'$ and $B'$ such that $\mathrm{diam}(A')\leq \mathrm{diam}(A)$, $\mathrm{diam}(B')\leq \mathrm{diam}(B)$, and $A'\cup B'=A\cup B.$
\end{theorem}

\begin{corollary}\label{perimeter}
	The construction in Theorem \ref{separacion de dos conjuntos2} verifies that
\begin{multline*}
	\mathrm{perimeter}(\mathrm{conv}(A))+\mathrm{perimeter}(\mathrm{conv}(B))\geq\\
 \mathrm{perimeter} (\mathrm{conv}(A'))+\mathrm{perimeter}(\mathrm{conv}(B')).
\end{multline*}
	If $\mathrm{conv}(A)\cap \mathrm{conv}($B$)\neq \emptyset$, then the inequality is strict.
\end{corollary}
\begin{proof}
We note $\mathrm{p}(S)$ to the perimeter of a set $S$. If $\mathrm{conv}(A)\cap \mathrm{conv}(B)$ is a segment or the empty set, there is nothing to prove. Let us assume that $\mathrm{conv}(A)\cap \mathrm{conv}(B)\neq \emptyset$. We note $l$ to the length of $L\cap \mathrm{conv}(A)\cap \mathrm{conv}(B)$, where $L$ is the splitting line of bad pairs from $A'$ and $B'$ in Theorem \ref{separacion de dos conjuntos2}. The following holds (see Figure \ref{convexAB}):
\begin{multline*}
 \mathrm{p} (\mathrm{conv}(A'))+\mathrm{p}(\mathrm{conv}(B'))\leq \\
 \sum_i \mathrm{p}(A_i)+ \sum_j \mathrm{p}(B_j) - \mathrm{p}(\mathrm{conv}(A)\cap \mathrm{conv}(B))+ 2l<\\
\sum_i \mathrm{p}(A_i)+ \sum_j \mathrm{p}(B_j)=
 \mathrm{p} (\mathrm{conv}(A))+\mathrm{p}(\mathrm{conv}(B))
 \end{multline*}
\end{proof}

\section{Some applications to clustering problems}\label{some applications}
From now on,  $S$ is a set of $n$ points in a normed plane $\mathbb{M}^2$. We assume that in our computation model the unit ball of $\mathbb{M}^2$ is given via an \emph{oracle} as it is described in Section 3.3 of \cite{Gr-Kl} or on page 316 in \cite{Mat}.

\subsection{2-clustering problem: minimize the maximum diameter.}% for the diameter with respect to the maximum}
\label{2-clustering problems}
Given a metric, the \textit{$2$-clustering problem of minimizing the maximum diameter} asks about how to split $S$ into two sets minimizing the maximum diameter. Avis solves the problem in $\mathbb{R}^2$  looking for two separable sets with the following algorithm ($O(n^2 \log^2 n)$ time).

 \begin{algorithm}

 \label{Avis}Given a set $S$ of $n$ points in the plane:
	\begin{enumerate}
	\item Sort the distances $d_i$ between the points of $S$ into increasing order ($O(n^2\log n)$ time).
	\item Locate the minimum $d_i$ that admits a \textit{stabbing line}\footnote{A \textit{stabbing line} for a set of segments is a line that intersects every segment of the set.} by a binary search. Use the graph $(S,E_{d_i})$, where $E_{d_i}$ is the set of edges meeting two points of $S$ at distance more than $d_i$, and the algorithm by Edelsbrunner et al. (\cite{E-M-P-R-W-W}) in order to find the stabbing line for $E_{d_i}$ ($O(m \log m)$ time each) as a subroutine.
	\end{enumerate}
\end{algorithm}
We obtain the following from Theorem  \ref{separacion de dos conjuntos2}.
\begin{corollary}
Given a set of $n$ points in $\mathbb{M}^2$, the $2$-clustering problem of minimizing the maximum diameter can be computed in $O(n^2 \log^2 n)$ time using Algorithm \ref{Avis}.
\end{corollary}
%\begin{proof}
% By Theorem  \ref{separacion de dos conjuntos2}, Algorithm \ref{Avis} works correctly taking  $O(n^2 \log^2 n)$ time in any general normed plane.
%\end{proof}.

Asano et al. (\cite{Asano}) reduce the cost of Algorithm \ref{Avis} to $O(n \log n)$ time in $\mathbb{E}^2$. They use the \textit{maximum spanning tree}\footnote{A \textit{maximum spanning tree} is a spanning tree whose total edge length is  as large as possible.} %Se considera un grafo en el que cada v\'{e}rtice es un punto de $S$ y  un eje es el segmento que une dos puntos. Un \textit{maximum spanning tree} es un grafo que cumple que 1) todos los puntos de $S$ est\'{a}n unidos entre s\'{\i} y por un solo camino (uni\'{o}n de ejes), no hay ciclos; 2) es m\'{a}ximo, es decir, la longitud total de los ejes es lo m\'{a}s larga posible}
 of $S$ (that can be constructed in such a time and space in $\mathbb{E}^2$; see \cite{MPSY}) instead of all the distances between points of $S$. This approach also works correctly in $\mathbb{M}^2$, but as far as we know, there is not a similar result about the cost of building a maximum spanning tree for any normed plane.

%His algorithm sorts the distances $d_i$ between the points of $S$ into increasing order, and then, the minimum $d_i$ that admits a \textit{stabbing line}\footnote{A stabbing line for a set of segments is a line that intersects every segment of the set} is located by a binary search. The algorithm by Edelsbrunner and others (\cite{E-M-P-R-W-W}) for finding the stabbing line for the set of $m$ line segments  in $O(m \log m)$ time (\cite{E-M-P-R-W-W}) is used as a subroutine.

%\begin{remark} In order to solve the above problem, Asano and others (\cite{Asano}) prove that it is enough to consider the \textit{maximum spanning tree}\footnote{Se considera un grafo en el que cada v'rtice es un punto de $S$ y  un eje es el segmento que une dos puntos. Un \textit{maximum spanning tree} es un grafo que cumple que 1) todos los puntos de $S$ est n unidos entre s!` y por un solo camino (uni¢n de ejes), no hay ciclos; 2) es m ximo, es decir, la longitud total de los ejes es lo m s larga posible} of $S$ instead of all the distances between points of $S$. This result does not depends on the metric. Then, the final time cost is $O(n \log n)$ time and $O(n)$ space for the Euclidean plane using  that the maximum spanning tree on $S$ can be constructed in such a time and space (\cite{MPSY}). This approach also works correctly in any normed plane, but as far as we know, there is not a similar result about the time and space cost for building a maximum spanning tree in a normed plane.
%\end{remark}	

\subsection{2-clustering problem: constraints over the diameters}\label{2-clustering constraints}
Given $d_1\geq d_2>0$,  % and $\mu(S)$ be a measure of $S$, for example its diameter, area or radius. For $d\geq d_2> 0$,
Hershberger and Suri (\cite{H-S}) solve in $\mathbb{E}^2$  the problem of dividing $S$ into two sets $S_1$ and $S_2$ such that $\operatorname{diam}(S_1)\leq d_1$ and $\operatorname{diam}(S_2)\leq d_2$ ($O(n\log n)$ time).
%The algorithm works in the following way:
%\begin{enumerate}
%	\item Compute $\mu(S)$
%	\item if $\mu(S)\leq d$ then
%	output $S_1=\emptyset, S_2=S$.
%	
%	 \item Else, let $a$ and $b$ be two points of $S$ such that $\mu(S)=\mu(\{a,b\})$.\\
%	 \hspace{1cm} this loop twice, once with $s_1:=a, s_2:=b$ and once with $s_1:=b, s_2:=a$.
%	 \begin{enumerate}
%	 	\item $S_1:=S\cap \bar{S}(s_2,d)$ and $S_2:=S\cap S(s_2,d)\cap \bar{S}(s_1,1)$.
%	 	\item $U:=S-(S_1\cup S_2)$.
%	 	\item Phase 1: it is a procedure such that a point $x$ of $U$ is assigned  to $S_2$ if there exists a point $s_1$ such that of $\|x-s_1\|>1$; and $x\in U$ is assigned to $S_1$ if there exists a point $s_2$ such that of $\|x-s_2\|>d$. Finally, the set $U$ only contains points of $S$ at distance less than or equal to 1 from any point of $S_1$ and less than or equal to $d$ from any point of $S_2$.
%	 	\item if $\mu(S_1)\leq 1$ and $\mu(S_2)\leq d$ then call Phase 2 and output $S_1,S_2.$ \\
%	 		 	Phase 2 is the bipartion of $U$ into $U_1$ and $U_2$ such that $\mu(U_1)\leq 1$ and $\mu(U_2)\leq d$. 	 	
%		 	 \end{enumerate}
%\item $S_1:=S_1\cup U_1$, $S_2:=S_2\cup U_2$. Then $\mu(S_1)\leq 1$ and $\mu(S_2)\leq d$ and the algorithm ends.	
%\end{enumerate}
%
% When $\mu$ is the diameter,
They uses the fact that if $\|a-b\|\geq d_1$, then $B(a,d_2)\cap B(b,d_1)$ can always be split into two subsets whose diameters are at most $d_1$ and $d_2$, respectively.
%, $S_1$ and $S_2$, such that $\mathrm{diam}(S_1)\leq d$ and  $\mathrm{diam}(S_2)\leq d_2$.
%
%\begin{lemma}\label{divide a set in two}
%	Let $a,b$ be two points in the $L_i$ plane ($i\in \{1,2,\infty\}$) and $\|a-b\|$=d. Let us consider the set  $S=B(a,1)\cap B(b,d)$. Then, we can always split $S$ into two subsets, $S_1$ and $S_2$, such that $\mathrm{diam}(S_1)\leq d$ and the $\mathrm{diam}(S_2)\leq 1$.
%\end{lemma}
%
%For proving Lemma \ref{divide a set in two} in the Euclidean case, Hersberger and Suri split $S$ into two subsets $S_1$ and $S_2$ using the sphere $S(p,1)$, where $p\in S(a,1)\cap S(b,d)$. But this strategy does not work even in some strictly convex planes (see Figure \ref{divide two subsets}).
%
%\begin{figure}[ht]
%	\begin{center}
%		\includegraphics[width=15cm]{divide_two_subsets.pdf}\\
%		\vspace*{-0.5cm}\caption{$\|P-U\|>d=\|A-B\|$}\label{divide two subsets}
%	\end{center}
%\end{figure}
%
%Indeed, Lemma \ref{divide a set in two}
Nevertheless, the following example shows that this can not be extended to $\mathbb{M}^2$. Let us consider $a=(0,0)$, $b=(-9.81,6.24)$, %$r=(-9.39,5.39)$, $s=(-8.24, 6.03)$,
and the strictly convex norm whose unit sphere is bounded by the two arcs of circles with center in $(0,10)$ and in $(0,-10)$, respectively, and radius $5\sqrt{13}$ (see Figure \ref{contraej_9_configuracion}). Let $\{r=(r_1,r_2), \quad s=(s_1,s_2)\}\in S(a,1)$ and $\{p,q\}=S(a,1)\cap S(b,1.1)$, such that $r_1=-9.39$, $r_2>0$, $s_1=-8.24$, $s_2>0$, and $p,r,s,q$ is the clockwise order on $S(a,1)$. It is verified that $\|a-b\|\geq 1.1$, $\mathrm{min}\{\|s-p\|,\| r-q\|,\|p-q\|\}> 1.1$ and  $\mathrm{min}\{\|r-p\|,\|s-q\|\}> 1$. Therefore the set $S=\{p,q,r,s\}\in B(a,1)\cap B(b,1.1)$ can not be divided in two subsets whose diameters are at most $1.1$ and $1$, respectively.
%\begin{figure}[ht]
%	\begin{center}
%		\hspace*{-1cm}\includegraphics[width=17cm]{contraej_15_configuracion.pdf}\\
%		\vspace*{-0.5cm}\caption{$S=\{p,q,r,s\}$ can not be divided in $S_1$ and $S_2$ such that $\mathrm{diam}(S_1)\leq 1.1$ and $\mathrm{diam}(S_2)\leq 1$}\label{contraej_9_configuracion}
%	\end{center}
%\end{figure}

\begin{figure}[ht]
	\vspace*{-2cm}
		\begin{center}
						\hspace*{-1.8cm}
		%\scalebox{.55}{\includegraphics*[778,63][93,484]{contraej_15_conletra_wedt_diego.ps}}\\
				\scalebox{.55}{\includegraphics[angle=270]{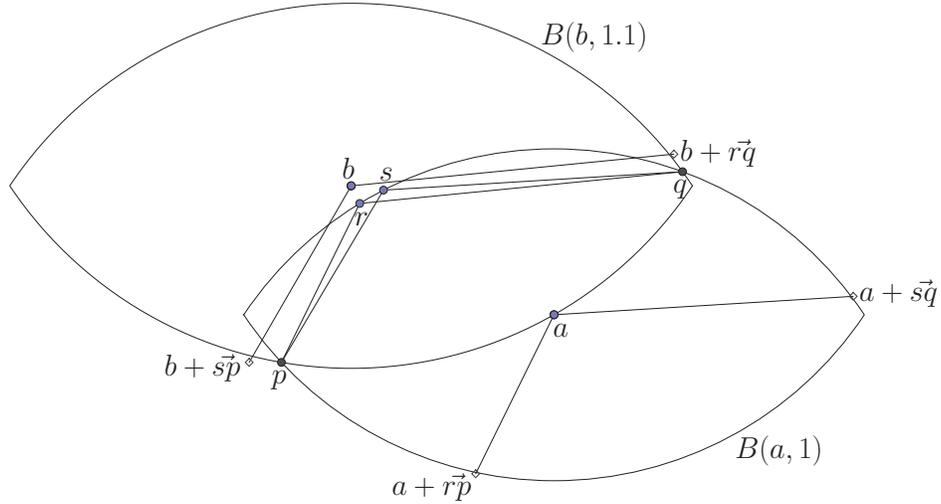}}\\
				\vspace*{-1.5cm}\caption{$S=\{p,q,r,s\}$ can not be divided in two subsets with diameters less than or equal to $1.1$ and $1$, respectively.}\label{contraej_9_configuracion}
	\end{center}
\end{figure}

However we can look for a separable pair of sets $S_1$ and $S_2$.
%\begin{algorithm}Given a set $S$ of $n$ points in the plane and $d\geq d_2>0$:\label{algorithm fixed radius}
%	\begin{enumerate}
%		\item Sort the distances between the points of $S$ into increasing order, %: $d_1\leq d_2\leq ...\leq d_{{n}\choose{2}}$,
%and build the graph $(S,E_d)$, where $E_d$ is the set of edges meeting two points of $S$ at distance more than $d$ ($O(n^2\log n)$ time).
%		\item Test if $E_d$ has a stabbing line ($O(n\log n)$ time with the algorithm presented in \cite{E-M-P-R-W-W}).
%		\item If the stabbing line does not exist, there is no solution (Theorem \ref{separacion de dos conjuntos2}). If the stabbing line exists, check if the minimum diameter of the separable sets $S_1$ and $S_2$ is less than or equal to $d_2$ ($O(n\log n)$ time).
%	\end{enumerate}
%\end{algorithm}

\begin{corollary}\label{algorithm fixed radius}
Given a set $S$ of $n$ points in $\mathbb{M}^2$, and $d_1\geq d_2>0$, the $2$-clustering problem of dividing $S$ into two sets $S_1$ and $S_2$ such that $\operatorname{diam}(S_1)\leq d_1$ and $\operatorname{diam}(S_2)\leq d_2$ can be solved in $O(n^2 \log n)$ time.
\end{corollary}
\begin{proof}
Let $E_{d_1}$ be the set of edges meeting two points of $S$ at distance more than $d_1$. Sort the distances between the points of $S$ into increasing order and build the graph $(S,E_{d_1})$ in $O(n^2\log n)$ time. Test if $E_{d_1}$ has a stabbing line (in $O(n\log n)$ time with the algorithm presented in \cite{E-M-P-R-W-W}). If the stabbing line does not exist, there is no solution (Theorem \ref{separacion de dos conjuntos2}). If the stabbing line exists, check if one of the subsets of $S$ separated by the the stabbing line has diameter less than or equal to $d_2$.
\end{proof}

\subsection{$k$-clustering problems}\label{$k$-clustering problems}
The \emph{$k$-clustering problem of minimizing the maximum diameter} is the natural extension of the case $k=2$ presented in Section \ref{2-clustering problems}. It is a particular case of the \emph{$k$-clustering problem of minimizing $\mathcal{F}$ over the diameters}, where $\mathcal{F}$ is a monotone increasing function $\mathcal{F}:\mathbb{R}^k\rightarrow\mathbb{R}$  that is applied over the diameters of the clusters (for instance, $\mathcal{F}$ can be the \emph{maximum}, the \emph{sum}, or the \emph{sum of squares} of the diameters). If we consider the radii instead of the diameters, we talk about the \emph{$k$-clustering problem of minimizing $\mathcal{F}$ over the radii}.
	
	The following result is presented in \cite{Capoyleas and others} for the Euclidean subcase.
	
	\begin{theorem}\label{radio}
		Let $S$ be a set of $n$ points in $\mathbb{M}^2$. Consider the $k$-clustering problem of minimizing a monotone increasing function $\mathcal{F}:\mathbb{R}^k\rightarrow \mathbb{R}$ that is applied over the diameters or over the radii of $k$ subsets of $S$. Then there is an optimal $k$-clustering such that each pair of clusters is linearly separable.
	\end{theorem}	
	\begin{proof}
Regarding the diameter, let us consider an optimal solution of the problem that minimizes the sum of the $k$ perimeters of the convex hulls of the clusters. Theorem \ref{separacion de dos conjuntos2} and Corollary \ref{perimeter} imply that there exist a $k$-clustering (with smaller or equal sum of perimeters) such that every pair of clusters are separable, and the value of $\mathcal{F}$ does not increase. %Since $\mathcal{F}$ is monotone increasing, these $k$ separable clusters are also an optimal solution for the  diameter.	
		
		Let us consider now the $k$-clustering problem of minimizing $\mathcal{F}$ over the radii. Let $C_1$ and $C_2$ be two clusters of $S$ from an optimal solution, and $B(u,r)$ and $B(v,r')$ be two minimal enclosing discs of $C_1$ and $C_2$, respectively, such that $C_1\subset B(u,r)$ and $C_2\subset B(v,r')$. If $S(u,r)\cap S(v,r')$ is the empty set or has only one connected component, $C_1$ and $C_2$ are separable. If $S(u,r)\cap S(v,r')$ has two different components $A_1$ and $A_2$, we consider a line $L$ meeting two points $p_1\in A_1$ and $p_2\in A_2$. % such that $L$ separates the centers of $D_1$ and $D_2$.
		Let $u_i=p_i-(v-u)$ and $v_i=p_i+(v-u)$ for $i=1,2.$ Let $S_1(u,r)$ be the part of $S(u,r)$ on the same side of the line $\langle p_1, p_2\rangle$ as $u_1$ and $u_2$; let $S_2(u,r)$ be the part of $S(u,r)$ on the  side of  $\langle p_1, p_2\rangle$ opposite to  $u_1$ and $u_2$. %; and similarly for $\gamma'$, $\gamma_1'$, and $\gamma_2'$.
		Let $S_1(v,r')$ be the part of $S(v,r')$ on the same side of the line $\langle p_1, p_2\rangle$ as $v_1$ and $v_2$; let $S_2(v,r')$ be the part of $S(v,r')$ on the  side of  $\langle p_1, p_2\rangle$ opposite to  $v_1$ and $v_2$. %; and similarly for $\gamma'$, $\gamma_1'$, and $\gamma_2'$.
		Then,  $S_2(u,r)\subseteq \operatorname{conv}(S_1(v,r'))$ and $S_2(v,r')\subseteq \operatorname{conv}(S_1(u,r))$ (see Gr\"{u}nbaum \cite{Grue1} and Banasiak \cite{Ban}).
		The subsets $S\cap \operatorname{conv}(S_1(u,r))$ and  $S\cap \operatorname{conv}(S_1(v,r'))$ are two separable clusters, and the minimal enclosing radius of the new clusters are no greater.
		
   Consequently we can reassign the points for every pair of intersecting clusters according to their position relative to the line $L$. Finally we obtain a $k$-clustering  such that every two clusters are separable and the value of $\mathcal{F}$ does not increase.
	\end{proof}

%	For a practical assignment of the points in the case of the radius, the so-called \textit{power diagram} (see Aurenhammer \cite{Au}) is proposed in \cite{Capoyleas and others}. This tesselation assures that every point $x$ is assigned to a disk $D_i$ with center $p_i$ and radius $r_i$,  such that $\|x-p_i\|^2-r_i^2$ is minimal, where $\|\cdot\|$ is the Euclidean norm.
%	
%	Given some points $p_i$ and radius $r_i$, Gavrilova \cite{Gavrilova} suggests some algorithms for computing some tesselations similar to power diagram but minimizing $\|x-p_i\|-r_i$ for the norm of the $L_{\infty}$ and $L_1$. She claims in \cite{Gavrilova} that some properties for these tesselations in various metrics (including Manhattan, supremum and Euclidean) were established in \cite{Gav}.
%	

	Therefore the optimal solution for the $k$-clustering problem of minimizing $\mathcal{F}$ over the diameter or over the radius is a planar dissection into $k$ convex polygonal regions,  such that each of them contains a cluster $C_i$. It can be represented by a graph $G=(V,E)$, where every vertex $v_i\in V$ corresponds to the region of a cluster $C_i$, and every edge $\{ij\}$ joints $v_i$ and $v_j$ if and only if a common boundary separates the polygonal regions that contain $C_i$ and $C_j$. The following algorithm by Capoyleas et al. (\cite{Capoyleas and others}) solves the $k$-clustering problem of minimizing a monotone increasing function $\mathcal{F}$ over the diameters or over the radii in the Euclidean plane.
	
	\begin{algorithm}

\label{algoritmo} Given a set $S$ of $n$ points in the plane:
		\begin{enumerate}
			\item For every graph (up to isometric ones) $G=\{V,E\}$ with $k$ vertices do the following:
			\item For every edge $\{ij\}\in E$, select a line and specify which side $H_{ij}$ of this line is to contain $C_i$ and which side $H_{ji}$ should contain $C_j$.
			\item For each point $p\in S$, determine to which side it belongs, and then for each $i$ evaluate $$R_i'=\bigcap_{\{ij\}\in E} H_{ij}$$
			Every region $R_i'$ contains $C_i$, and they are pairwise disjoint (see Lemma 8 in \cite{Capoyleas and others}). If  each point happens to fall into exactly one cluster, we have a candidate for an optimal solution.
			\item Evaluate the diameter (or the radius) of every cluster $C_i$, and then the function $\mathcal{F}$.
			\item Take the minimum of the values of $\mathcal{F}$.
		\end{enumerate}
	\end{algorithm}

	\begin{corollary}\label{algorithm clustering}
		Let $S$ be a set of $n$ points in $\mathbb{M}^2$. For any fixed $k$, the geometric $k$-clustering problem of minimizing a monotone increasing function $\mathcal{F}$ over the diameters or over the radii  is solvable by Algorithm \ref{algoritmo}. It takes polynomial time for the diameter.
	\end{corollary}
\begin{proof}
By Theorem \ref{radio}, Algorithm \ref{algoritmo} (see Lemma 8 and Theorem 9 in  \cite{Capoyleas and others} for details) works correctly in %when $\mathcal{F}$ is applied over the diameters or over the the radii  in
$\mathbb{M}^2$ too.

The number of non-isometric graphs with $k$ vertices is fixed. The number of edges is at most $3k-6$, and $n$ points can be separated by these edges in $O(n^{6k-12})$ different ways.	
	Regarding step (4) of Algorithm \ref{algoritmo}, the diameter of a set of $n$ points can be computed in $O(n \log n)$ time in $\mathbb{M}^2$ with the same algorithm that in $\mathbb{E}^2$ (\cite{P-S}). Therefore, the $k$-clustering problem for minimizing the diameter in $\mathbb{M}^2$ is solvable in polynomial time.
\end{proof}

It seems that there is not an optimal solution for determining the minimal enclosing radius of a set of points in $\mathbb{M}^2$. Two algorithms are presented in \cite{J} for strictly convex normed planes. The first one is similar to Elzinga/Hearn's and takes $\Omega(n^2)$ time. The other is similar to Shamos/Hoey's and  enables an $O(n)$ search for the optimal disk once the  farthest-point Voronoi diagram of the set is constructed. Nevertheless, the strictly convex case can be solved by an easier way because the radius and the covering circle of each cluster are determined by at most three points (\cite{Al-Ma-Sp}, \cite{Al-Ma-Sp2}). Hence it would be enough to check only $O(n^{3k})$ possibilities.%, which is a better bound than the bound of Algorithm \ref{algoritmo}.

\subsection{$3$-clustering problems}\label{3-clustering}
%As it is said in Section \ref{2-clustering problems}, the $2$-clustering problem for the diameter with respect to the maximum can be solved in any general normed plane by the algorithms presented by Avis (\cite{Av}) and by Asano and others (\cite{Asano}).
Having in mind Theorem \ref{radio},
we can do the following in order to solve the \emph{$3$-clustering problem minimizing the maximum diameter}:
	%\begin{enumerate}
(1) Separate the $n$ points in all possible two linear separable sets ($O(n^4)$  possibilities); (2) Use Algortihm \ref{Avis} to split the second of these sets; (3) Determine the optimal solution.		
	%\end{enumerate}
	This takes $O(n^6 \log^2 n)$ time if Avis' approach is used, and it could be improved with the algorithm by Asano et al.
But we prove in this section that Hauger-Rote's $3$-clustering approach for $\mathbb{E}^2$ (\cite{Hagauer-Rote}) works correctly in $\mathbb{M}^2$ with some modifications.

We fix a normal basis $\{x,y\}$ in $\mathbb{M}^2$ such that $x$ is Birkhoff orthogonal to $y$ (namely, such that $\|x\|\leq \|x+\lambda y\|$ for every $\lambda\in \mathbb{R}$). % In order to do this, it is only necessary to  justify some results presented bellow.
It is assumed that two given points of $S$ have different $x$ and $y$ coordinate (the points are rotated if it is necessary).
Given $d>0$, %and having in mind Theorem \ref{radio},
the algorithm searches all the possible linearly separable subsets $A,B,C,$ such that the maximum diameter is less than or equal to $d$. The point $a\in S$ with minimum $x$-coordinate is placed in $A$, and each point $a'\in S$ such as $\|a-a'\|\leq d$ is tested as the possible point of $A$  with the maximum $x$-coordinate. Any $u\in S\cap \overline{aa'}$ is assing to $A.$ The plane is divided in the following three zones by the lines $\langle a,a' \rangle$ and $a'+\beta y$ ($\beta\in \mathbb{R}$):
$$\begin{array}{ll}
\textsc{North}:=&\{u\in S/ \; u=\alpha a+(1-\alpha) a'+\beta y, \text{ with }1>\alpha>0,\beta> 0\}\\%(u_x,u_y) \text{ is "above" the segment } \overline{aa'}\}\\
\textsc{\textsc{South}}:=&\{u\in S/ \;u=\alpha a+(1-\alpha) a'+\beta y, \text{ with }1>\alpha>0,\beta<0\}\\%(u_x,u_y)
%\{(u_x,u_y)\in S/ u_x<a'_x, \text{and } (u_x,u_y) \text{ is "below" the segment } \overline{aa'}\}\\
\textsc{\textsc{East}}:=&\{u\in S/ \; u=-\alpha a+(1+\alpha) a'+\beta y, \text{ with }\alpha>1\}\\%(u_x,u_y)\}.
\end{array}
$$

There is not any point of $S$ on the "left" of $a+\beta y$ ($\beta\in \mathbb{R}$). \textsc{East} contains the points of $S$ on the "right" of the  line $a'+\beta y$. The points of $S$ on the left of $a'+\beta y$ are contained either in \textsc{\textsc{North}} (if they are "above" $\overline{aa'}$) or in \textsc{South} (if they are "bellow" $\overline{aa'}$).

%We have a valid solution if the subsets $S_1,S_2,S_3$ are separable, $\operatorname{diam}(S_i)\leq d$ ($i=\{1,2,3\}$), $a\in S_1$ is the point with the minimum $x$-coordinate in $S$,  and $a'\in S_1$ is the point with the maximum $x$-coordinate in $S_1$.
%Hagauer and Rote explain %in \cite{Hagauer-Rote}how to find
 Solutions are tested in three different cases: Case 1, $\text{\textsc{North}}\subseteq A$; Case 2, $\text{\textsc{South}}\subseteq A$; and  Case 3, $\text{\textsc{North}}$ and $\text{\textsc{South}}$ are not completely contained in $A$.
We note  $A_{cand}$ to the set of points that \textit{could} be placed in $A$ for every candidate $a'$:
$$A_{cand}=S\cap B(a,d)\cap B(a',d).$$ %  \{u\in S\;/\; \|a-u\|\leq d, \; \|a'-u\|\leq d\}.$$

\begin{lemma}\label{north-south} With the previous notations, the following holds in $\mathbb{M}^2$:
		$$\operatorname{diam}(A_{cand}\cap \operatorname{\textsc{North}})\leq d \hspace{0.3cm} \text{and}\hspace{0.3cm} \operatorname{diam}(A_{cand}\cap \operatorname{\textsc{South}})\leq d.$$
	\end{lemma}
\begin{proof}
Proposition 1(iii) in \cite{Capoyleas and others} for $\mathbb{E}^2$ can be applied for $\mathbb{M}^2$: due to the geometry of the figure, $\mathrm{diam}(A_{cand}\cap \operatorname{\textsc{\textsc{North}}})$ is equal to the distance between two support lines, and one of them has to pass through $a$ or $a'$. Since all the points are within $B(a,d)\cap B(a',d)$, the diameter is at most $d$. Similarly for $\mathrm{diam}(A_{cand}\cap \operatorname{\textsc{\textsc{South}}})$.
\end{proof}

\begin{lemma}\label{north-south2}
 Let us assume the following conditions in $\mathbb{M}^2$:
 \begin{itemize}
 \item $\max\{\operatorname{diam}(A),\operatorname{diam}(B),\operatorname{diam}(C)\}\leq d$,
 \item $A,B,C$ are separable,
 \item $B\cap \text{\textsc{North}}\neq \varnothing$ and  $C\cap \text{\textsc{\textsc{South}}}\neq \varnothing.$
 \end{itemize}
 If there exist a pair of points $u=(u_x,u_y),v=(v_x,v_y)\in \text{\textsc{East}}$ such that $\|u-v\|> d$ and $u_y>v_y$, then $u\in B$ and $v\in C.$
\end{lemma}
\begin{proof}
Since $\|u-v\|>d$, the points $u$ and $v$ can not be situated in the same subset of the partition $A,B,C$. We can choose $u'=(u'_x,u'_y)\in B\cap \text{\textsc{North}}$ and $v'=(v'_x,v'_y)\in C\cap \text{\textsc{South}}$.

%\begin{figure}
%	\vspace*{2cm}
%	\scalebox{2}{\includegraphics[bb= 146 69 265 107]{twocircles6.pdf}}\\
%	%\vspace{-3cm}
%	\caption{Intersection of two  closed convex curves.}\label{intersection two circles}
%\end{figure}

Let us assume that $u_y>v'_y$. If $v$ is situated in the shaded zone in Figure \ref{contraej_17_1}, $v$ must belong to $C$, because in other case either the pair of segments $\overline{vu'}$ and $\overline{uv'}$ or the pair of the segments $\overline{vu'}$ and $\overline{aa'}$ cross.

If $v$ is not situated in the shaded zone in Figure \ref{contraej_17_1} and $u_x<v_x$ (for instance, $v=v_1$ in Figure \ref{contraej_17_1}), we consider the two intersection points of the line $u+\lambda y$ with the line $v+\lambda x$ and with the line $v+\lambda (u'-v)$, that we note by $\bar{u}$ and $\widetilde{u}$, respectively. Since $x$ is Birkhoff orthogonal to $y$, $u+\lambda y$ supports $S(v,\|\bar{u}-v\|)$ on $\bar{u}$, and  $\|v-u'\|\geq\|v-\widetilde{u}\|\geq \|v-u\|\geq \|v-\bar{u}\|$. As a result of $\|v-u'\|\geq \|v-u\|> d$, $v\in C$.%Therefore, $\|v-u'\|>d$ and $v\in C$.
\begin{figure}[ht]
	\begin{center}
		%\scalebox{.9}{\includegraphics*[221, 60][517, 247]{contraej_17_3.eps}}
		%\scalebox{.9}{\includegraphics[width=16cm]{contraej_17_3.pdf}}
		\scalebox{1.1}{\includegraphics[bb=228 90 550 241, clip]{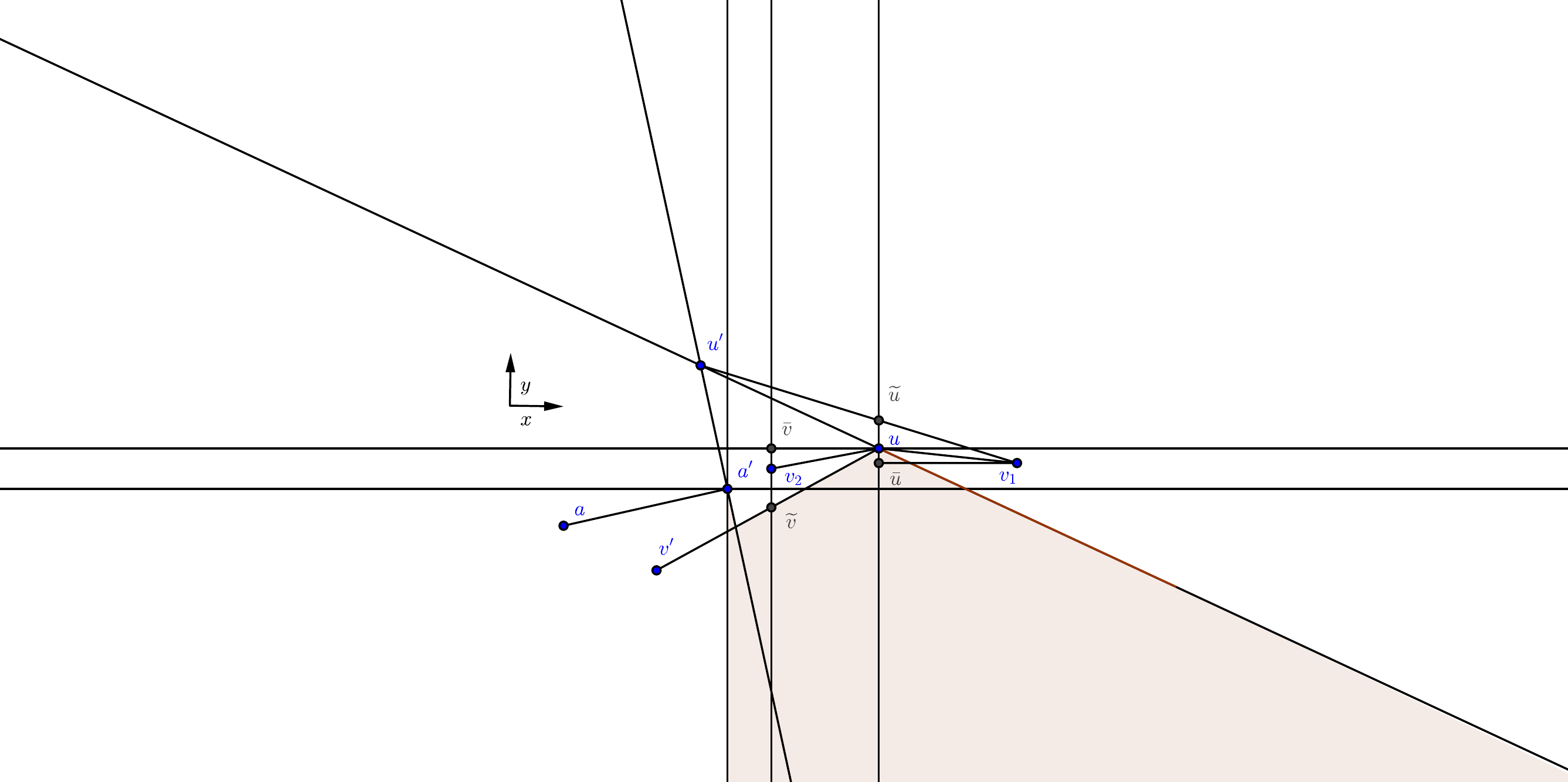}}
		\caption{}
		\label{contraej_17_1}
	\end{center}
\end{figure}

If $v$ is not situated in the shaded zone in Figure \ref{contraej_17_1} and  $u_x>v_x$ (for instance, $v=v_2$ in Figure \ref{contraej_17_1}), we consider the two intersection points of the line $v+\lambda y$ with the line $u+\lambda x$ and with the line $u+\lambda (u-v')$, that we note by $\bar{v}$ and $\widetilde{v}$, respectively. Since $x$ is Birkhoff orthogonal to $y$, the line $v+\lambda y$ is the support line of $S(u,\|u-\bar{v}\|)$ on $\bar{v}$, and %the sphere with center in $u$ and radius $\|u-v\|$ either separates these two intersection points or both are on the sphere. In any case, $\|u-v'\|\geq\|u-v\|>d$ and $u\in B$.
$\|u-v'\|\geq \|u-\widetilde{v}\|\geq \|u-v\|\geq \|u-\bar{v}\|$.
% $\|u-\bar{v}\|\leq \|u-v\|\leq \|u-\widetilde{v}\|\leq \|u-v'\|$.
As a result of $\|u-v'\|\geq \|u-v\|> d$, $u\in B$.

The analysis is similar if $u_y<v'_y$.
\end{proof}

The algorithm of Hagauer and Rote works in the following way.

\begin{algorithm}

\label{algorithmHR}
	Fix $a$ with the minimum $x$-coordinate. Then, for every $a'\in S$:
\begin{enumerate}
\item Calculate \textsc{North}, \textsc{South} and \textsc{East} ($O(n)$ time).
%\item Test Case 1 ($\mathrm{\textsc{North}}\subseteq A$), Case 2 ($\mathrm{\textsc{South}}\subseteq A$), or Case 3.
%\begin{enumerate}
\item Test Case 1 ($\text{\textsc{North}}\subseteq A$). Note $H=\{a,a',\mathrm{\textsc{North}}\}$ and check if $\mathrm{diam}(H)\leq d$ ($O(n\log n)$ time). If yes, define
$$
A:=H\cup \left(\mathrm{\textsc{South}}\cap \big(\cap_{x\in H} B(x,d)\big)\right),
$$
 Obtain $B$ and $C$ solving a $2$-clustering problem for the set $S\setminus A$ (for instance, in $O(n^2 \log^2 n)$ time by Algorithm \ref{Avis}).
\item Test Case 2 and manage it in a similar way as Case 1 ($O(n^2 \log^2 n)$ time).
\item Test Case 3 (neither \textsc{North} nor \textsc{South} are completely contained in $A$). Assign the points of $S$ that are initially forced to be in $A,B,C$ (by Lemma \ref{north-south2} and the rest of conditions)  to the initial sets $A_0,B_0,C_0$:\\
$\begin{array}{ll}
A_0:=&\{a,a'\}\\
B_0:=&\{u\in \mathrm{\textsc{North}}/\; u\notin \cap_{x\in A_0}B(x,d)\} \\
     &\cup \{u\in \mathrm{\textsc{East}}/\; \exists v\in \mathrm{\textsc{East}}, \|u-v\|>d, u_y>v_y\}\\
C_0:=&\{u\in \mathrm{\textsc{South}}/\;  u\notin \cap_{x\in A_0}B(x,d)\}\\
     & \cup \{u\in \mathrm{\textsc{East}}/\; \exists v\in \mathrm{\textsc{East}}, \|u-v\|>d, u_y<v_y\},
     \end{array}$\\
and the rest of the points of $S$ to one of the following candidate sets:\\
$\begin{array}{ll}
AB_{cand}:=&\mathrm{\textsc{North}}\setminus B_0\\
CA_{cand}:=&\mathrm{\textsc{\textsc{South}}}\setminus C_0\\
BC_{cand}:=&\mathrm{\textsc{East}}\setminus ({B_0\cup C_0}).
     \end{array}$\\
Stop if $B_0$ and $C_0$ are not disjoint (there is not a solution by  Lemma \ref{north-south2}). In other case, assign the points of the candidate sets to $A,B,C$ by Hauger and Rote's procedure \emph{distribute} (see \cite{Hagauer-Rote}).
\end{enumerate}
\end{algorithm}

The set $A$ defined in Case 1 (as well as in Case 2) is the uniquely maximal feasible set with diameter  less than or equal to $d$ (Lemma \ref{north-south}).
The procedure \emph{distribute} does not depend on the metric, therefore one solution is found in Case 3 (if there exists).

%or $\text{\textsc{\textsc{South}}}\subseteq S_2$, the approach presented in \cite{Hagauer-Rote} can be applied using Lemma \ref{north-south}. There is a maximal feasible set $S_1$ uniquely defined and we only have to solve finally a $2$-clustering problem by some of the algorithms presented in Section \ref{2-clustering problems}.

In order to implement the algorithm in $\mathbb{E}^2$, Hauger and Rote use
the $d$-\emph{ball hull} (also called $d$-\emph{circular hull}) of a set $S$, and the data structure introduced by Hershberger and Suri (\cite{H-S}).
We justify in the Appendix that this data structure can be used in $\mathbb{M}^2$ (see Proposition \ref{HS-structure}).

\begin{theorem}\label{tripartition}
Given a set of $n$ points in $\mathbb{M}^2$ and $d>0$, we can determine with the Algorithm \ref{algorithmHR} whether there is a
partition of $S$ into sets $A,B,C$ with diameters at most $d$. This can be done in $O(n^3 \log^2 n)$ time.
\end{theorem}
\begin{proof}
Hagauer-Rote's proof (for $\mathbb{E}^2$) of the first part of statement  depends on some lemmas  (see Lemma 3 to Lemma 6 in \cite{Hagauer-Rote}) and on Theorem \ref{separacion de dos conjuntos}.
%Theorem \ref{separacion de dos conjuntos} joint with Lemma 3 to Lemma 6  justify in \cite{Hagauer-Rote} that Algorithm \ref{algorithmHR} works correctly in the Euclidean plane.
Once Theorem \ref{separacion de dos conjuntos}, and Lemma 3 and Lemma 4 in \cite{Hagauer-Rote} are extended to $\mathbb{M}^2$ by our Theorem \ref{separacion de dos conjuntos2}, Lemma \ref{north-south} and Lemma \ref{north-south2}, respectively, the rest of the lemmas and proofs can be applied to any normed plane. Regarding the complexity of the algorithm, the Hershberger and Suri's data structure  can be managed (see Proposition \ref{HS-structure} in Appendix). Hence, for every $a'\in S$ Case 1 and Case 2 take  $O(n^2 \log^2 n)$ time (using Algorithm \ref{Avis} as a subroutine), and Case 3 takes $O(n \log n)$ time as in $\mathbb{E}^2$ (see \cite{Hagauer-Rote} for details). Therefore, the $3$-clustering algorithm takes $O(n^3 \log^2 n)$ time.
\end{proof}
Finally, a binary search on the ${n \choose 2}$
 distances occurring in $S$  combined with Theorem \ref{tripartition}  solves the optimization problem.
\begin{theorem}
Given a set of $n$ points in $\mathbb{M}^2$, we can  construct in $O(n^3 \log^3 n)$ time a
partition of $S$ into sets $A,B,C$ such that the largest of the three dia-meters is as small as possible.
\end{theorem}

%\begin{theorem}[\cite{H-S}, Theorem 4.17]\label{Theorem 4.17}
%Given a set $S$ of $n$ points, we can compute a data structure in time $O(n \log n)$ such that (1) intersections of $\operatorname{bh}(S,\lambda)$ with a circle of radius $d\geq \lambda$ can be found in worst-case time $O(\log n)$, and (2) the data structure can be updated after a point delection in amortized time $O(\log n)$.
%\end{theorem}

%We say that an arc meeting two points is a $d$-minimal arc if it is situated on a sphere of radius $d$ and is separated from the center by the line meeting the points. The following result is proved in \cite{Ma-Ma-Sp2}.

\section*{Appendix}\label{Appendix}
Given a set $S$ in $\mathbb{M}^2$ and $d>0$, the $d$-\emph{ball hull}  $\operatorname{bh}(S,d)$ of $S$ (also called $d$-\emph{circular hull}) is the intersection of all the balls of radius $d$ and center $x\in \mathbb{M}^2$ that contain $S$:
$$
\operatorname{bh}(S,d)=\bigcap_{S\subseteq B(x,d)}B(x,d).
$$
%The same structure is also valid in a strictly convex normed plane.
The data structure introduced by Hershberger and Suri  (\cite{H-S}) orders the points of the input set $S$ by their $x$-coordinates, and  situates them on the leaves of a complete binary tree %\footnote{A \emph{complete binary tree} is a  tree in which 1) every node has at most two children, 2) every level, except possibly the last, is completely filled, and 3) all nodes are as far left as possible.}
$T(S)$. Every node of $T(S)$ represents the $d$-ball hull of the points in the leaves of its subtree. Therefore, the root of $T(S)$ represents the $d$-ball hull of $S$. The information about every node is stored like a doubly linked list of its vertices such that for every vertex the predecessor and the successor is known. Since a point can be the vertex of more than one ball hull, for economizing space every point is only stored as vertex at the highest level in the tree at which it appears on a ball hull.  It is proved (\cite{H-S}, see Lemma 4.1 to Lemma 4.16, and Theorem 4.17) that the data structure $T(S)$ (therefore the ball hull of $S$) can be built initially in $O(n \log n)$ time and it supports the following operations $(1)$ and $(2)$ in $\mathbb{E}^2$:
\begin{enumerate}
\item Given a query point $u\in S$, determine in $O(\log n)$ a point $v\in S$ such that $\|u-v\|\geq d$, if such a point exists.
\item It can be updated after a point deletion in $O(\log n)$ time.
\end{enumerate}

%Let $\mathbb{M}^2$ be a normed plane.  %as it is described in Section 3.3 of \cite{Gr-Kl} or on page 316 in \cite{Mat}.
%The computation of the intersection of two balls is a basic operation in our approach where $B$ is given via an ``oracle''.

The intersection of two spheres in $\mathbb{M}^2$ is always the  union of two segments, each of which may degenerate to a point or to the empty set (\cite{Grue1}, \cite{Ban}; see also \cite[$\S$ 3.3]{MSW}). As a consequence, it is obtained the following (\cite{Ma-Ma}).
\begin{lemma}\label{3.0extended}
	Given $d>0$, for every pair of points $p,q \in \mathbb{M}^2$ whose distance is less than or equal to $2d$, there exist two circular arcs of radius  $d$ meeting them (eventually only one, if they degenerate to the same segment) which belong to every disc of radius $d$ containing $p$ and $q$. These two arcs (if they are really two) are situated in different half planes  bounded by the line $\langle p, q \rangle$. The center of each disc defining these two minimal arcs is an extreme points of the segments $S(p,d)\cap S(q,d)$.
\end{lemma}

We call \textit{$d$-minimal arc meeting $p$ and $q$} to each of  these arcs cited in Lemma \ref{3.0extended}.

\begin{lemma}
	Given $d'\geq d>0$,  every ball  of radius $d$ in $\mathbb{M}^2$ contains every $d'$-minimal arc meeting two points of the ball.
	\end{lemma}	
	\begin{proof}Let us consider $p,q\in B(u,d)$ for some $u\in \mathbb{M}^2.$ For every $r>0$,   the $r$-ball hull of the set $\{p,q\}$ is the set bounded by the two $r$-minimal arcs meeting $p$ and $q$ (Lemma \ref{3.0extended}).  Since the ball hull operator is decreasing with respect to the radius (\cite{Ma-Ma-Sp}), then $\mathrm{bh}(\{p,q\}, d')\subseteq \mathrm{bh}(\{p,q\}, d)$ and the statements holds.
\end{proof}
The following lemma  (\cite{Ma-Ma}) describes the geometry of the ball hull of a finite set in $\mathbb{M}^2$, and it is very similar to the Euclidean subcase.
\begin{lemma}\label{con1}
	Let $S=\{p_1,p_2,\dots,p_n\}$ be a finite set in $\mathbb{M}^2$. Then
	\[
	\operatorname{bh}(S,d)=\bigcap_{S\subset B(x_s,d)}B(x_s,d)=\operatorname{conv}(\bigcup_{i,j=1}^n \widehat{p_ip_j}),
	\]
	where $x_s$ are some extreme points of the components $S(p_i,d)\cap S(p_j,d)$, and $\widehat{p_ip_j}$ are $d$-minimal arcs  meeting points of $S$ and whose centers are these extreme points $x_s$.
\end{lemma}

 %The intersection of two balls is a basic operation in our approach for any normed plane.

 All the proofs from Lemma 4.1 to Lemma 4.16 and Theorem 4.17 in \cite{H-S} can be extended almost word by word\footnote{If the norm is not strictly convex, the intersection of two balls could contain a segment. Regarding the extension of some statements of \cite{H-S}, every intersection segment must be computed only as one intersection point.} to $\mathbb{M}^2$ using the notion of minimal arc and Lemmas \ref{3.0extended} to \ref{con1}.
  As a consequence, Hershberger and Suri's data structure works in a normed plane as does in $\mathbb{E}^2$.

\begin{proposition}\label{HS-structure}
The structure for managing ball hulls in $\mathbb{E}^2$ described by Hershberger and Suri  works correctly in $\mathbb{M}^2$ and with the same time cost. It can be built initially in $O(n \log n)$ and supports the following operation:
\begin{enumerate}
\item Given a query point $u\in S$, determine in $O(\log n)$ a point $v\in S$ such that $\|u-v\|\geq d$, if such a point exists.
\item It can be updated after a point deletion in $O(\log n)$ time.
\end{enumerate}
\end{proposition}
%\begin{proof}
% It is proved (\cite{H-S}, Lemma 4.1 to Lemma 4.16, Theorem 4.17) that the data structure $T(S)$ (therefore the circular hull of $S$) can be built initially in $O(n \log n)$ time and it supports operations $(1)$ and $(2)$ in the Euclidean plane.
%
% The intersection of two balls is a basic operation in our approach for any normed plane. If the norm is not strictly convex, the intersection of two balls can be a segment when both balls are tangent. In such a case,  regarding the number of intersection points, this segment can be computed only one time. Having this in mind joint with the notion of minimal arcs and Lemma \ref{3.0extended} to \ref{con1}, all the proofs of Lemma 4.1-Lemma 4.16 and Theorem 4.17 in \cite{H-S} can be reproduced almost word by word to any normed plane.  As a consequence, Hershberger and Suri's data structure works in such as normed planes as does in the Euclidean plane.
%\end{proof}
	\nocite{*}
	\bibliographystyle{amsplain}

\end{document}